\documentclass{llncs}
\usepackage{fullpage}

\usepackage{epsfig}
\usepackage{graphics}
\usepackage{amsmath,setspace,amssymb}

\usepackage[dvipsnames,usenames]{color}
\usepackage[colorlinks=true,urlcolor=Blue,citecolor=Green,linkcolor=BrickRed]{hyperref}
\begin{document}

\title{Improved Submatrix Maximum Queries in Monge Matrices}

\author{
Pawe{\l} Gawrychowski \inst{1} \and Shay Mozes\inst{2}\thanks{Mozes and Weimann  supported in part by Israel Science Foundation grant 794/13.} \and Oren Weimann\inst{3}$^{\ast}$
}
\institute{
MPII, \href{mailto:gawry@mpi-inf.mpg.de}{gawry@mpi-inf.mpg.de}  \and
IDC Herzliya, \href{mailto:smozes@idc.ac.il}{smozes@idc.ac.il}   \and
University of Haifa, \href{mailto:oren@cs.haifa.ac.il}{oren@cs.haifa.ac.il}  
}

\date{}
\maketitle

\begin{abstract}
We present efficient data structures for submatrix maximum
queries in Monge matrices and Monge partial matrices.
For  $n\times n$  Monge matrices, we give a data structure that requires
$O(n)$ space and
answers submatrix maximum queries in $O(\log n)$ time. The best
previous data
structure [Kaplan et al., SODA`12] required $O(n \log n)$ space and $O(\log^2 n)$ query time. 
We also give an alternative data structure with constant query-time and $ O(n^{1+\varepsilon})$ construction time
and space for any fixed $\varepsilon<1$. 
For $n\times n$ {\em partial} Monge matrices we obtain a data
structure with $O(n)$ space  and  $O(\log n \cdot \alpha(n))$ query time.
The data structure of Kaplan et al. required  $O(n \log n \cdot \alpha(n))$ space and $O(\log^2 n)$ query time.

\medskip

Our improvements are enabled by a technique for exploiting the structure of the upper
envelope of Monge matrices to efficiently report column maxima in skewed
rectangular Monge matrices. We hope this technique will  be
useful in obtaining faster search algorithms in Monge
partial matrices.
In addition, we give a linear upper bound on
the number of breakpoints in the upper envelope of a Monge partial matrix. This shows that the inverse Ackermann $\alpha(n)$
factor in the analysis of the data structure of Kaplan et. al is superfluous. 

\end{abstract}

\section{Introduction}

A matrix $M$ is a {\em Monge} matrix if for any pair of rows $i<j$ and columns $k<\ell$ we have that $M_{ik}+ M_{j\ell} \ge M_{i\ell }+ M_{jk}$. 
Monge matrices have many
applications in combinatorial optimization and computational
geometry. For example, they arise in problems involving distances in
the plane~\cite{Hoffman61,Monge,Park91,Tiskin}, and in problems on convex $n$-gons~\cite{SMAWK,AggarwalP88}. 
See~\cite{BKR96} for a survey on Monge matrices and their uses in combinatorial
optimization.

In this paper we consider the following problem: Given an $n \times n$ Monge matrix $M$, construct a data structure that  can report the maximum entry in any query submatrix (defined by a set of consecutive rows and a set of consecutive columns).
Recently, Kaplan, Mozes, Nussbaum and
Sharir~\cite{KaplanMozesNussbaumSharir} presented an $\tilde O(n)$ space\footnote{The $\tilde O(\cdot)$ notation
  hides polylogarithmic factors in $n$.} data structure with $\tilde O(n)$ construction time and $O(\log^2 n)$ query time. They also described an
extension of the data structure to handle {\em partial} Monge matrices (where some of the
entries of $M$ are undefined, but the defined entries in each row
and in each column are
contiguous). The
extended data structure incurs larger polylogarithmic factors in the
space and construction time. 
Both the original and the extended data structures have various important applications. They are
used in algorithms that efficiently find the largest empty rectangle containing a
query point, in dynamic distance oracles for planar graphs,
and in algorithms for maximum flow in planar graphs~\cite{BKMNW11}.
See~\cite{KaplanMozesNussbaumSharir} for more details on the
history of this problem and its applications.

Note that, even though explicitly representing the input matrix
requires $N=\Theta(n^2)$ space, the additional space required by the
submatrix maximum data structure of~\cite{KaplanMozesNussbaumSharir} is only $\tilde O(n)$. In many applications (in particular~\cite{BKMNW11,KaplanMozesNussbaumSharir}), the matrix $M$ is not stored explicitly but any entry of $M$ can be computed when needed in $O(1)$ time. The space required by the application  is therefore dominated by the size of the submatrix maximum data structure. 
With the increasing size of problem instances, and with current memory
and cache architectures, space 
often becomes the most significant resource. 

For general (i.e., not Monge) matrices, a long line of research over the last three decades including~\cite{AFL07,CR1989,DemaineLandauWeimann,GBT84,YuanA10} achieved $\tilde O(N)$ space and $\tilde O(1)$ query data structures, 
culminating with the $O(N)$-space  $O(1)$-query data structure of Yuan
and Atallah~\cite{YuanA10}. 
Here $N=n^2$ denotes the total number of entries in the matrix.  
It is also known~\cite{BrodalDR10} that reducing the space to $O(N/c)$ incurs an $ \Omega(c)$ query-time. Tradeoffs requiring $O(N/c)$ additional space and $\tilde O(c)$ query-time were  given in~\cite{BrodalESA,BrodalDR10}.
When the matrix has only $N=o(n^2)$ nonzero entries,
the problem is known in computational geometry as the {\em orthogonal range searching} problem  on the $n \times n$ grid. In this case as well, various tradeoffs with $\tilde O(N)$-space and $\tilde O(1)$-query appear in a long history of results including
~\cite{AlstrupEtAl,Patrascu,Chazelle88,Munro,GBT84}. In particular, a
linear $O(N)$-space data structure was given by
Chazelle~\cite{Chazelle88} at the cost of an $O(\log^\varepsilon n)$
query time.  See~\cite{Nekrich} for a survey on orthogonal range search.
 
\paragraph{\bf Contribution.} 
Our first contribution is in designing $O(n)$-space $O(\log n)$-query data structures for
submatrix maximum queries in Monge matrices and in partial
Monge matrices (see Section~\ref{sec:linearspace}). Our data
structures improve upon the data structures of Kaplan et al. in both space and query time. 
Consequently, using our data structures for 
finding the largest empty rectangle containing a query point improves the space and query time by logarithmic factors. 

We further provide alternative data structures with
faster query-time; We achieve $O(1)$ query-time at the cost of
$ O(n^{1+\varepsilon})$ construction time and space for an arbitrarily small constant
$0<\varepsilon<1$ (see Section~\ref{sec:fast-query}).
 
Our results are achieved by devising a data structure for reporting column
maxima in $m\times n$ Monge matrices with many more columns than rows ($n >>
m$). We refer to this data structure as the \emph{micro} data structure. The
space required by the micro data structure
is linear in $m$, and independent of $n$. Its construction-time
depends only logarithmically on $n$. 
The query-time is $O(\log \log n)$, the time required
for a predecessor query in a set of integers bounded by $n$. 
We use the micro data structure in the
design of our submatrix maximum query data structures, exploiting its 
sublinear dependency on $n$, and an ability to
trade off construction and query times.

For partial Monge matrices, we provide a tight $O(m)$ upper bound on the complexity of the upper
envelope 
(see Section~\ref{sec:upperbound}). The best
previously known bound~\cite{SA} was $m \alpha(m)$, where $\alpha(m)$ is the
 inverse Ackermann function. This upper bound immediately
implies that the $\alpha(m)$ factor stated in the space and construction
time of the data structures of Kaplan et al. is superfluous.

Notice that the upper envelope of a {\em full} $m\times n$ Monge
matrix also has complexity $O(m)$.  The famous SMAWK
algorithm~\cite{SMAWK} can find all column maxima in $O(n+m)$ time. 
However, this is not the case for partial Monge matrices.
Even for simple partial Monge matrices such as triangular, or
staircase matrices, where it has been known for a long time that the complexity
of the upper envelope is linear, the fastest known algorithm for finding all
column maxima is the $O(n \alpha(m) + m)$ time algorithm of Klawe
and Kleitman~\cite{KlaweK90}. 
We hope that our micro data structure will prove useful for obtaining
a linear-time algorithm. 
The known algorithms, including the $(n \alpha(m) + m)$-time
algorithm of Klawe and Kleitman~\cite{KlaweK90}, partition the matrix
into skewed rectangular matrices, and use the SMAWK algorithm. It is
plausible that our micro data structure will yield a speed up since
it is adapted to skewed matrices.

\section{Preliminaries and Our Results}\label{preliminaries}
In this section we overview the data structures
of~\cite{KaplanMozesNussbaumSharir} and highlight our results.

A matrix $M$ is a {\em Monge} matrix if for any pair of rows $i<j$ and columns $k<\ell$ we have that $M_{ik}+ M_{j\ell} \ge M_{i\ell }+ M_{jk}$. 
A matrix $M$ is {\em totally monotone in columns} if for any pair of
rows $i<j$ and columns $k<\ell$ we have that if $M_{ik}\le
M_{jk}$ then $M_{i\ell }\le M_{j\ell}$. 
Similarly, $M$ is {\em totally monotone in rows} if for any pair of rows $i<j$ and columns $k<\ell$  we have that if $M_{ik}\le M_{i\ell}$ then $M_{jk}\le M_{j\ell}$.  Notice that the Monge property implies total monotonicity (in columns and in rows) but the converse is not true. When we simply say  {\em totally monotone} (or TM) we mean totally monotone  {\em in columns} (our results symmetrically apply to totally monotone  {\em in rows}).

A matrix $M$ is a \emph{partial} matrix if some 
entries of $M$ are undefined, but the defined entries in each row
and in each column are
contiguous.
We assume w.l.o.g. that every row has at least one defined element and that the defined elements form a single connected component (i.e., the defined column intervals in each pair of consecutive rows overlap). If this is not the case then only minor changes are needed in our algorithms. 
A partial TM (resp., Monge) matrix is a partial matrix whose defined entries
satisfy the TM (resp., Monge)
condition.
 
\noindent The following propositions are easy to verify:

\begin{proposition}
\label{prop:adjacent condition}
An $m\times n$ matrix $M$ is Monge iff $M_{i,j}+M_{i+1, j+1} \ge M_{i+1,j}+M_{i,j+1}$ for all $i=1,2,\ldots,m-1$ and $j=1,2,\ldots,n-1$.
\end{proposition}

\begin{proposition}
\label{prop:undef}
If a matrix $M$ is partial Monge, then it remains partial Monge after replacing any element of $M$ by a blank, so long as the defined (non-blank) entries in each row and in each column remain contiguous.
\end{proposition}

\begin{proposition}
\label{prop:duplicate}
If an $m$-by-$n$ matrix $M$ is (partial) Monge, then the $(m+1)$-by-$n$ matrix resulting by replacing any row of $M$ by two identical copies of that row is also (partial) Monge. An analogous statement holds for duplicating any column of $M$. 
\end{proposition}

We consider $m \times n$ matrices, but for simplicity we sometimes state the
results for $n \times n$ matrices. 
For a Monge matrix $M$, denote $r(j)=i$ if the  maximum element in
column $j$ lies in row $i$. (We assume this maximum element is unique.
It is simple to break ties by, say, taking the highest index.)
The {\em upper envelope} $\mathcal{E}$ of all the rows of
$M$ consists of the $n$ values $r(1), \ldots, r(n)$. Since $M$ is Monge we have that $r(1)\le r(2)\le  \ldots \le r(n)$ and so $\mathcal{E}$ can be implicitly represented in $O(m)$ space by keeping only the $r(j)$s of $O(m)$ columns called {\em breakpoints}. Breakpoints are the columns~$j$ where $r(j)\ne r(j+1)$. 
The maximum element $r(\pi)$ of any column $\pi$ can then be retrieved
in $O(\log m)$ time by  a binary search for the first
breakpoint-column $j$ after $\pi$, and  setting $r(\pi)=r(j)$.

The first data structure of~\cite{KaplanMozesNussbaumSharir} is a balanced binary
tree $T_h$ over the rows of $M$. A node $u$ whose subtree contains $k$
leaves (i.e., $k$ rows) stores the $O(k)$ breakpoints of the $k \times
n$ matrix $M^u$  defined by these $k$ rows and all columns of $M$.
A leaf represents a single row and
requires no computation. An internal node $u$ obtains its breakpoints by merging the breakpoints of its two children: its left child $u_1$ and its right $u_2$. 
By the Monge property, the list of
 breakpoints of $u$ starts with a prefix of  breakpoints of
$u_1$ and ends with a suffix of breakpoints of
$u_2$. 
Between these there is possibly one new breakpoint $j$. The prefix and suffix parts can be found easily in $O(k)$ time by linearly comparing the lists of breakpoints of $u_1$ and $u_2$. The new breakpoint $j$ can then be found in additional $O(\log n)$ time via binary search. Summing $O(k + \log n)$ over all 
nodes
of $T_h$ gives $O(m (\log m + \log n))$  time. 
The
total size of $T_h$ is $O(m \log m)$.
 
Note that the above holds even if $M$ is not Monge but only TM.
This gives rise to a 
data structure that answers subcolumn (as opposed to  submatrix) queries:

\medskip

\noindent {\bf Subcolumn queries in TM matrices~\cite{KaplanMozesNussbaumSharir}.} {\em 
Given a $n\times n$ TM  matrix,  one can construct, in $O(n \log  n)$
time, a data structure of size $O(n \log n)$ that reports the maximum
in a query
column and a contiguous range of rows in $O(\log n)$~time. }

\medskip
\noindent The maximum entry in a query column $\pi$ and a contiguous range of rows $R$ is found using $T_h$ by identifying $O(\log m)$ \emph{canonical nodes} of $T_h$.
A node $u$ is canonical if $u$'s set of rows is contained in $R$ but the set of
rows of $u$'s parent is not.
For each such canonical node $u$, we find in $O(\log m)$ time the maximum element in column $\pi$ amongst all the rows of $u$.
The output is the largest of these and  the  total query time is
$O(\log^2 m)$. The query time  can be reduced to
$O(\log m)$ by using fractional cascading~\cite{CGfrac}. 

The first results of our paper improve the above subcolumn query data structure of~\cite{KaplanMozesNussbaumSharir}, as indicated in Table~\ref{table} under subcolumn query in  TM matrices. 
 The next data structure of~\cite{KaplanMozesNussbaumSharir} extends the 
queries from subcolumn to  submatrix  (specified by ranges $R$ of consecutive rows, and $C$ of consecutive columns.)

\medskip
\noindent {\bf Submatrix queries in Monge matrices~\cite{KaplanMozesNussbaumSharir}.} {\em 
Given a $n\times n$ Monge matrix,  one can construct, in $O(n\log  n)$
time, a data structure of size $O(n \log n )$ that  reports the
maximum entry 
in a query
submatrix in $O(\log^2 n))$ time. }

\medskip
\noindent To obtain $O(\log^2 n)=O(\log m(\log m+\log n))$ query time, note that $R$ is the disjoint union of $O(\log m)$ canonical nodes of $T_h$. 
For each such canonical node $u$, we use $u$'s list of breakpoints $\{j_1,j_2,\ldots, j_k\}$ to  find in $O(\log m+\log n)$ time the maximum element in all rows of $u$ and the range of columns $C$. This is done as follows:
we first identify in $O(\log m)$ time the set $\mathcal{I}=\{j_a,j_{a+1},\ldots, j_b\}$
of $u$'s breakpoints that are fully contained in $C$. The columns of $C$ that are to the left of $j_a$ all have their maximum element in row $r(j_a)$. To find the maximum of these we 
construct, in addition to $T_h$, a symmetric binary tree $\mathcal B$ that can report in $O(\log n)$ time the maximum entry in a query {\em row} and a contiguous range of {\em columns}. $\mathcal B$ is built in $O(n(\log  m+\log n))$ time and $O(n \log n)$ space using the subcolumn query data structure on the transpose of $M$. This is possible since $M$ is Monge.\footnote{In fact it 
suffices that $M$ is a TM matrix whose transpose is also TM.}
Similarly, we find in $O(\log n)$ time the maximum in all columns of $C$ that are to the right of  $j_b$. 

To find the maximum in all columns between $j_a$ and $j_b$, 
let $m(j_i)$ denote the maximum element in the columns interval
$(j_{i-1},j_i]$ (note it must be in row $r(j_i)$). 
We wish to find $\max\{ m(j_{a+1}),\ldots, m(j_b)\}$ which corresponds
to a Range Maximum Query in the array $A^u = \{m(j_1), \ldots,
m(j_k)\}$. We compute the array $A^u$ (along with a naive RMQ data structure with logarithmic query time) of every node $u$ during the construction of $T_h$. Most of the entries of $A^u$ are simply copied from $u$'s children arrays $A^{u_1}$ and $A^{u_2}$. The only new $m(\cdot)$ value that $u$ needs to compute is for the single new breakpoint $j$ (that is between the prefix from $u_1$ and the suffix from $u_2$). Since $m(j)$ must be in row $r(j)$ it can be  computed in $O(\log n)$ time by a single query to $\mathcal B$.
 
Overall, we get a query time of $O(\log m + \log n)$ per canonical node $u$ for a total of $O(\log m(\log m + \log n))$.  Building $T_h$ (along with all the RMQ arrays $A^u$) and $\mathcal B$ takes total 
$O((m+n)(\log m+ \log n))$ time and  $O(m\log m+n\log n)$ space. Our two improvements to this bound  of~\cite{KaplanMozesNussbaumSharir}  are stated in Table~\ref{table} under submatrix queries in Monge matrices.

 The next data structures of~\cite{KaplanMozesNussbaumSharir} extend
 the above subcolumn and submatrix data structures from full to {\em
   partial}  TM matrices. The construction is very
 similar. Merging
the breakpoints of the two children $u_1$, $u_2$ of a node $u$
of $T_h$ is slightly more involved now, since the envelopes may cross
each other multiple times. 
The number of breakpoints of any subset of consecutive $k$ rows is $O(k \cdot \alpha(k))$~\cite{SA}, and so there are
$O(m \log m \cdot \alpha(m)  )$ breakpoints in total over all nodes of $T_h$ (as opposed to $O(m)$ in full matrices). This implies the following

\medskip

\noindent {\bf Subcolumn queries in partial TM matrices~\cite{KaplanMozesNussbaumSharir}.} {\em 
 Given a partial TM $n\times n$ matrix,  one can construct, in $O(n
 \log^2 n \cdot \alpha(n) )$ time, a data structure of size $O(n  \log
 n \cdot \alpha(n))$ that reports the maximum entry 
in a query
column and a contiguous range of rows in $O(\log n)$ time. }

\medskip 

\noindent We improve this data structure to the same bounds we get for full matrices. i.e, we show that our bounds for full  matrices also apply to partial  matrices. This is stated in Table~\ref{table} under  subcolumn query in Partial TM matrices.
Finally,~\cite{KaplanMozesNussbaumSharir} extended their submatrix data structure  from full to partial Monge matrices. 
It uses a similar construction of $T_h$ and $\mathcal{B}$ as in the case of full matrices, but again requires the additional $O(\log m \cdot \alpha(m) +\log n \cdot \alpha(n))$ multiplicative  factor to store the breakpoints 
of all nodes of $T_h$ and $\mathcal B$.


\begin{table}[h]\small
\hskip1cm\begin{tabular}{|l|l|l|l|l|l|}
\hline
 property &  query type & space & construction time & query time & \\
  \noalign{\global\arrayrulewidth0.05cm}
\hline 
 \noalign{\global\arrayrulewidth0.4pt} 
TM & subcolumn & $O(n \log n)$ & $O(n \log n)$ &
$O(\log n)$ &  Lemma 3.1 in~\cite{KaplanMozesNussbaumSharir} \\ \hline

 TM & subcolumn & $O(n)$ & $O(n\log n/\log\log n)$ &
$O(\log n)$ & Lemma~\ref{lemma:3.1improvement} here \\ \hline

 TM & subcolumn & $O(n^{1+\varepsilon})$ & $O(n^{1+\varepsilon})$ &
$O(1)$ & Lemma~\ref{lemma:3.1fast-query} here\\ \hline 

 \noalign{\global\arrayrulewidth0.05cm}
\hline 
 \noalign{\global\arrayrulewidth0.4pt}

Monge & submatrix & $O(n \log n)$ & $O(n\log n)$ &
$O(\log^2 n)$ & Theorem 3.2 in~\cite{KaplanMozesNussbaumSharir} \\ \hline

Monge & submatrix & $O(n)$ & $O(n\log n)$ &
$O(\log n)$ & Theorem~\ref{theorem:3.2improvement} here  \\ \hline

Monge & submatrix & $O(n)$ & $O(n\log n/\log\log n)$ &
$O(\log^{1+\varepsilon} n)$ & Corollary~\ref{cor:3.2improvement} here  \\ \hline

Monge & submatrix & $O(n^{1+\varepsilon})$ & $O(n^{1+\varepsilon})$ &
$O(1)$ & Theorem~\ref{theorem:3.2fast-query} here   \\ \hline 

 \noalign{\global\arrayrulewidth0.05cm}
\hline 
 \noalign{\global\arrayrulewidth0.4pt} 
Partial TM & subcolumn &$O(n\log n \cdot \alpha(n))$ & $O(n\log ^2 n \cdot
\alpha(n))$ &
$O(\log n)$ &Lemma 3.3 in~\cite{KaplanMozesNussbaumSharir} \\
\hline

 Partial TM & subcolumn & $O(n)$ & $O(n\log n/\log\log n)$ &
$O(\log n)$ & Lemma~\ref{lemma:filltheblanks} here \\ \hline

 Partial TM & subcolumn & $O(n^{1+\varepsilon})$ & $O(n^{1+\varepsilon})$ &
$O(1)$ & Lemma~\ref{lemma:filltheblanks} here\\ \hline 

 \noalign{\global\arrayrulewidth0.05cm}
\hline 
 \noalign{\global\arrayrulewidth0.4pt}

Partial Monge & submatrix & $O(n\log n \cdot \alpha(n))$ & $O(n\log ^2 n \cdot
\alpha(n))$ &
$O(\log^2 n)$ &Theorem 3.4 in~\cite{KaplanMozesNussbaumSharir} \\
\hline

Partial Monge & submatrix & $O(n)$ & $O(n\log n)$ &
$O(\log n \cdot \alpha(n))$ & Theorem~\ref{theorem:3.4improvement} here \\
\hline

Partial Monge & submatrix & $O(n)$ & $O(n\log n / \log\log n)$ &
$O(\log^{1+\varepsilon} n \cdot \alpha(n))$ & Corollary~\ref{cor:3.4improvement} here \\
\hline

\end{tabular}

  \caption{Our results compared to~\cite{KaplanMozesNussbaumSharir}. }
\label{table}
\end{table}


\medskip
\noindent {\bf Submatrix queries in partial Monge matrices~\cite{KaplanMozesNussbaumSharir}.} {\em 
Given a $n\times n$ partial Monge matrix,  one can construct, in
$O(n\alpha(n)  \log^2 n )$ time, a data structure of size $O(n
\alpha(n) \log n)$ that reports the maximum entry 
in a query
submatrix in  $O(\log^2 n)$~time. }

\medskip

\noindent We remove the $O(\log n \cdot \alpha(n))$ multiplicative
factor and obtain the bounds stated in the bottom of
Table~\ref{table}. The $\alpha(n)$ factor is removed by showing that 
the number of breakpoints in the upper envelope
of a partial  Monge matrix is linear. 

\section{Linear-Space Data Structures}\label{sec:linearspace}
In this section we present our data structures that improve the space
to $O(n)$ and the query time to $O(\log n)$.
We begin by introducing a new data structure for the case where a query is composed of an {\em entire} column (as opposed to a range of rows). This new data structure (which we call the micro data structure) is designed to work well when the number of rows in the matrix is much smaller than the number of columns.
We denote by $pred(x,n) = O(\min\{\log x, \log\log n\})$ the time to query a predecessor data structure with $x$ elements from $\{1,\ldots,n\}$.

\begin{lemma} [the micro data structure]\label{lemma:micro}
Given a $x\times n$ TM matrix and $r >0$, one can construct in $O(x \log n /\log r)$ time, a data structure of size $O(x)$ that given a query column can report the maximum entry in the entire column in 
$O(r+pred(x,n))$ time. 
\end{lemma}
\begin{proof}
Out of all $n$ columns of the input matrix $M$, we will designate
$O(x)$ columns as {\em special} columns. For each of these special columns we will eventually compute its maximum element.  
The first $x$ special columns of $M$ are  columns $1, n/x, 2n/x, 3n/x,\ldots, n$ and are denoted $j_1,\ldots, j_x$. 

Let $X$ denote the $x \times x$ submatrix obtained by taking all $x$
rows but only the $x$ special columns $j_1,\ldots, j_x$. It is easy to verify that $X$
is TM. We can therefore run the SMAWK
algorithm~\cite{SMAWK} on $X$ in  $O(x)$ time and obtain the column
maxima of all special columns. 
Let $r(j)$ denote the row containing the maximum element in column $j$.
Since $M$ is TM, the $r(j)$ values are monotonically
non-decreasing. Consequently, $r(j)$ of a non-special column $j$ must be
between $r(j_i)$ and $r(j_{i+1})$ where $j_i<j$ and $j_{i+1}>j$ are
the two special  columns bracketing $j$  
 (see Figure~\ref{fig}). 
 
\begin{figure}[h!]
   \centering
   \includegraphics[scale=0.3]{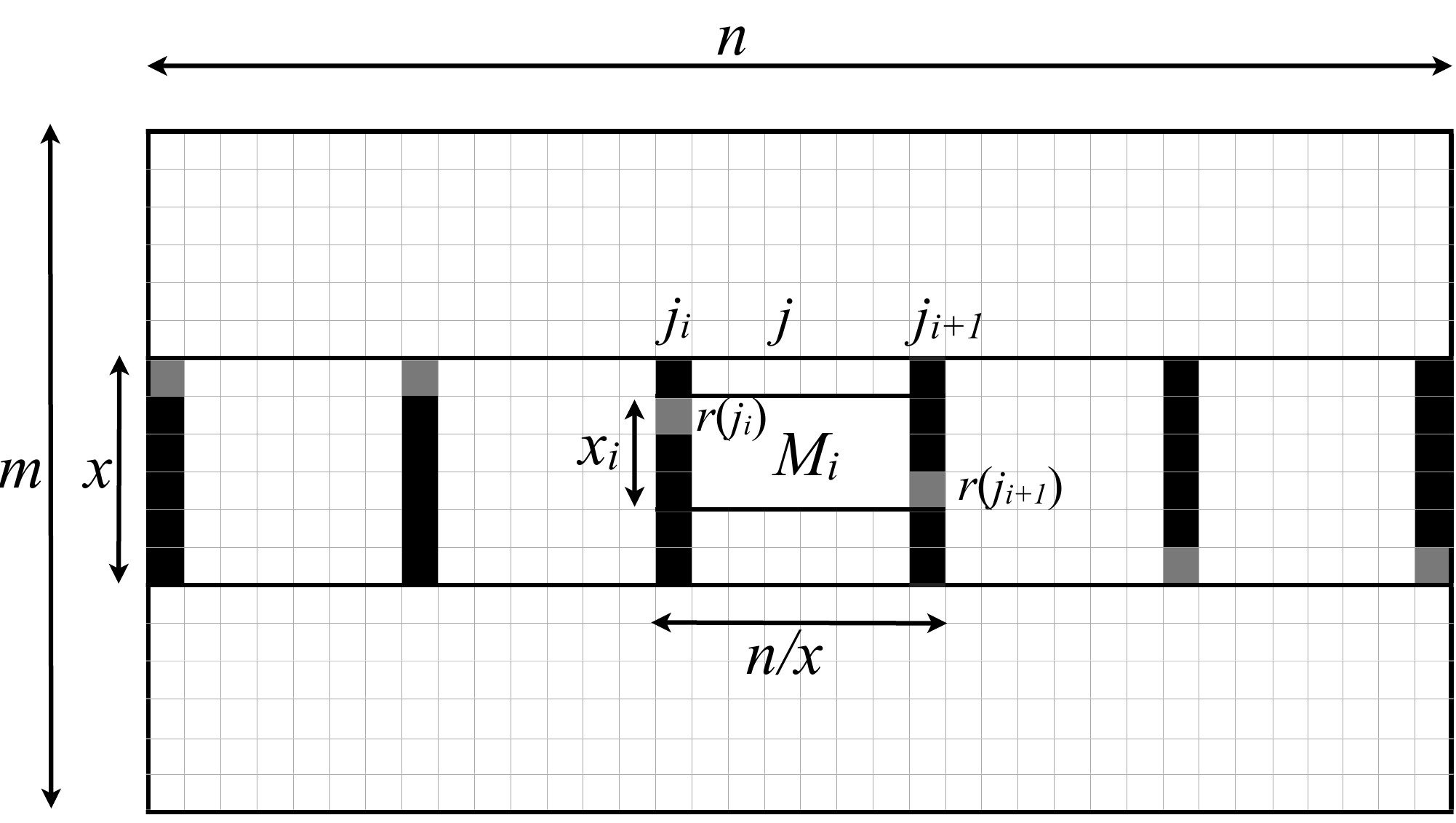}
   \caption{An $x\times n$ matrix inside an $m\times n$ matrix. The
     black columns are the first $x$ special columns. The
     (monotonically  non-decreasing) gray cells inside these special
     columns are the column maxima (i.e., the $r(j_i)$ values of breakpoints $j_i$). The maximum element of column $j$ in the $x\times n$ matrix must be between $r(j_i)$ and $r(j_{i+1})$ (i.e., in matrix $M_i$).}
  \label{fig}
 \end{figure}

For every $i$, let $x_i = r(j_{i+1}) - r(j_i)$. If $x_i  \le r$ then
{\em no} column between $j_i$ and $j_{i+1}$ will ever be a special
column. When we will query such a column $j$ we can simply check
(at query-time) the $r$ elements of $j$ between rows $r(j_i)$ and
$r(j_{i+1})$ in $O(r)$ time. If, however, $x_i  > r$, then we designate
more special columns between $j_i$ and $j_{i+1}$. This is done
recursively on the $x_i  \times (n/x)$ matrix $M_i$ composed of
rows $r(j_i),\ldots, r(j_{i+1})$ and columns $j_i,\ldots,
j_{i+1}$. That is, we mark $x_i$ evenly-spread columns of $M_i$ as
special columns,  and run SMAWK in $O(x_i)$ time on the $x_i \times
x_i$ submatrix $X_i$   obtained by taking all $x_i$ rows but only
these $x_i$ special columns. We continue recursively until either $x_i
\le r$ or the number of columns in $M_i$ is at most $r$. In the
latter case, before terminating, the recursive call runs SMAWK in $O(x_i+r)=O(x_i)$ time on the $x_i
\times r$ submatrix $X_i$ obtained by taking the $x_i$ rows and {\em
  all} columns of $M_i$ (i.e., all columns of $M_i$ will become
special). 

After the recursion terminates, every column $j$ of $M$ is either
special (in which case we computed its maximum), or its maximum  is
known to be in one of at most $r$ rows (these rows are specified by the $r(\cdot)$
values of the two special columns bracketing $j$). 
Let $s$ denote the total number of columns that are marked as special. 
We claim that $s = O(x \log n /\log r)$. To see this, notice that the
number of columns in every recursive call decreases by a factor of at
least $r$ and so the recursion depth is $O(\log_r n) = O(\log n /\log
r)$. In every recursive level, the number of added special columns is
$\sum x_i$ over all $x_i's$ in this level that are at least $r$. In
every recursive level, this sum is bounded by $2x$ because each one of
the $x$ rows of $M$ can appear in at most two $M_i$'s  (as the last
row of one and the first row of the other). Overall, we get $2x \cdot
O(\log n /\log r) = O(x \log n /\log r)$. 

Notice that $s = O(x \log n /\log r)$ implies that the total time
complexity of the above procedure is also $O(x \log n /\log r)$. This
is because whenever we run SMAWK on a $y \times y$ matrix it takes
$O(y)$ time and $y$ new columns are marked as special.  To complete
the construction, we go over the $s$ special columns from left to right in
$O(s)$ time and throw away (mark as non-special) any column whose
$r(\cdot)$ value is the same as that of the preceding special column. This
way we are left with only $O(x)$ special columns, and the difference in
$r(\cdot)$ between consecutive special columns is at least $1$ and at most
$r$. In fact, it is easy to maintain $O(x)$ (and not $O(s)$) space
{\em during} the construction by only recursing on sub matrices $M_i$
where $x_i >1$. We note that when $r=1$, the eventual special columns are
exactly the set of breakpoints of the input matrix $M$.

The final data structure is a predecessor data structure that holds
the $O(x)$ special columns and their associated $r(\cdot)$
values. Upon query of some column $j$, we search in $pred(x,n)$ time
for the predecessor and successor of $j$ and obtain the two $r(\cdot)$
values. We then  search for the maximum of column $j$ by explicitly
checking all the (at most $r$) relevant rows of column $j$. The query time is
therefore  $O(r+pred(x,n))$ and the space $O(x)$.  
\qed \end{proof}

\paragraph{\bf A linear-space subcolumn  data structure.} 

\begin{lemma}\label{lemma:3.1improvement}
Given a $m\times n$ TM matrix, one can construct, in $O(m(\log n+
\log m)/\log\log m)$ time, a data structure of size $O(m)$ that can report the maximum entry in a query
column and a contiguous range of rows in $O(\log m)$ time.
\end{lemma}

\begin{proof}
Given an $m\times n$ input matrix $M$ we partition it into $m/x$ matrices $M^1, M^2, \ldots, M^{m/x}$ where $x=\log m$. 
Every $M^i$ is an $x\times n$ matrix composed of $x$ consecutive rows
of $M$. We construct  the micro data structure of
Lemma~\ref{lemma:micro} for each $M^i$ separately choosing
$r=x^\varepsilon$ for any constant $0\!<\!\varepsilon\!<\!1$. This requires
$O(x \log n /\log r) = O(x \log n /\log x) $ construction time per
$M^i$ for a total of $O(m \log n /\log \log m)$ time. We obtain a
(micro) data structure of total size $O(m)$ that upon query $(i,j)$
can report in  $O(x^\varepsilon+pred(x,n)) = O(\log^\varepsilon m)$
time the maximum entry in column $j$ of $M^i$.

Now, consider the $ (m/x) \times n$ matrix $M'$, where $M'_{ij}$ is
the maximum entry in column $j$ of $M^i$. We cannot afford to store
$M'$ explicitly, however, using the micro data structure we can
retrieve any entry $M'_{ij}$ in $O(\log^\varepsilon m)$ time. We next
show that $M'$ is also TM. 

 For any pair of rows $i<j$ and any pair of columns $k<\ell$ we need to show that if $M'_{ik}\le M'_{jk}$ then $M'_{i\ell }\le M'_{j\ell}$. Suppose that $M'_{ik}, M'_{jk}, M'_{i\ell }$, and $M'_{j\ell}$ correspond to entries $M_{ak}, M_{bk}, M_{c\ell }$, and $M_{d\ell}$ respectively. We assume that $M_{ak}\le M_{bk}$ and we need to show that  $M_{c\ell }\le M_{d\ell}$. Notice that $M_{ck} \le M_{ak}$ because $M_{ak}$ is the maximal  entry in column $k$ of $M^i$ and $M_{ck}$ is also an entry in column $k$ of $M^i$. Since $M_{ck} \le M_{ak}$ and $M_{ak}\le M_{bk}$ we have that $M_{ck} \le M_{bk}$. 
Since $M_{ck} \le M_{bk}$, from the total monotonicity of $M$, we have  that  $M_{c\ell} \le M_{b\ell}$. Finally, we have  $M_{b\ell} \le M_{d\ell}$ because $M_{d\ell}$ is the maximal  entry in column $\ell$ of $M^j$ and $M_{b\ell}$ is also an entry in column $\ell$ of $M^j$. We conclude that $M_{c\ell} \le M_{d\ell}$. 

Now that we have established that the matrix $M'$ is TM, we can
use the subcolumn data structure of~\cite{KaplanMozesNussbaumSharir} (see previous section) on $M'$. Whenever an entry $M'_{ij}$ is desired, we can retrieve it using the micro data structure. This gives us the macro data structure: it is of size $O(m/x\cdot \log(m/x)) = O(m)$ and  can report in $O(\log m)$ time the maximum entry of $M'$ in a query column and a contiguous range of rows. 
It is
built in $O(m/x \cdot (\log(m/x) +\log  n)\cdot x^\varepsilon  )$ time which is  $O(m(\log  n+
\log m)/\log\log m)$  for any choice of $\varepsilon<1$. 

To complete the proof of Lemma~\ref{lemma:3.1improvement} we need to show how to answer a general query in $O(\log m)$ time. Recall that a query is composed of a column of $M$ and a contiguous range of rows. If the range is smaller than $\log m$ we can simply check all elements explicitly  in $O(\log m)$ time and return the maximum one. Otherwise, the range is composed of three parts: a prefix part of length at most  $\log m$, an infix part that corresponds to a range in $M'$, and a suffix part of length at most  $\log m$. The prefix and suffix are computed explicitly in $O(\log m)$ time. The infix is computed by querying the macro data structure in $O(\log m)$ time.
\qed \end{proof}
\paragraph{\bf A linear-space submatrix  data structure.} 

\begin{theorem}\label{theorem:3.2improvement}
Given a $m\times n$ Monge matrix,  one can construct, in $O((m+n)(\log  n+
\log m))$ time, a data structure of size $O(m+n)$ that can report the maximum entry in a query
submatrix in $O(\log m+\log n)$ time. 
\end{theorem}

\begin{proof}
Recall from Section~\ref{preliminaries} that the  submatrix data structure of~\cite{KaplanMozesNussbaumSharir} 
 is composed of the tree $T_h$ over the rows of $M$ and the tree
 $\mathcal B$ over the columns of $M$. Every node $u\in T_h$ stores
 its breakpoints along with the RMQ array $A^u$ (where $A^u[j]$ holds
 the value of the maximum element between the $(j-1)$'th and the
 $j$'th breakpoints of $u$). If $u$ has $k$ breakpoints then they are
 computed along with $A^u$ in $O(k+\log n)$ time: $O(k)$  to copy from
 the children of $u$ and $O(\log n)$  to find the new breakpoint and
 to query $\mathcal B$. As opposed to~\cite{KaplanMozesNussbaumSharir}, we don't use a naive RMQ data structure but instead one of the existing linear-construction constant-query RMQ data structures such as~\cite{HT84}.

To prove Theorem~\ref{theorem:3.2improvement} we begin with two  changes to the above. First, we build $T_h$ on the rows of the $ (m/x)  \times n$ matrix $M'$ instead of the $m \times n$ matrix $M$ (again, when an entry  $M'_{ij}$ is desired, we retrieve it using the micro data structure in $O(x^\varepsilon)$ time). Second, for $\mathcal B$ we use the data structure of Lemma~\ref{lemma:3.1improvement} applied to the transpose of $M$. $\mathcal B$'s construction requires $O(n(\log  m+ \log n)/\log\log n)$ time and $O(n)$ space. After this, constructing $T_h$ (along with the $A_u$ arrays) on $M'$ requires $O(m/x\cdot \log(m/x))=O(m)$ space and $O((m/x) (\log(m/x) + \log  n)\cdot x^\varepsilon)=O(m(\log  m+ \log n)/\log\log m)$ time 
by choosing $x=\log m$ and  any $\varepsilon<1$.

Finally, we construct a data structure $T_v$  that is symmetric to $T_h$ but applied to the transpose of $M$. Notice that  $T_v$ is built on the columns of an $m \times (n/\log n)$ matrix $M''$ instead of the $m \times n$ matrix $M$.  
The construction of $T_v$, from a symmetric argument to the previous paragraph, also takes $O((m+n)(\log  n+
\log m)/\log \log m)$ time and $O(m+n)$ space. 

We now describe how to answer a submatrix query with row range $R$ and
column range $C$. Let $R'$ be the set of consecutive rows of $M'$ whose
corresponding rows in $M$ are entirely contained in $R$. 
Let $R_{p}$ be the prefix of $O(\log m)$ rows of $R$ that do not correspond to rows of $R'$. 
Let $R_{s}$ be the suffix of $O(\log m)$ rows of $R$ that do not correspond to
rows of $R'$.
We define the subranges $C',C_{p},C_{s}$ similarly (with respect to columns and to $M''$).
The submatrix query $(R,C)$ can be covered by the following:
(1) a submatrix query $(R',C)$ in $M'$, (2) a submatrix query $(R,C')$
in $M''$, and (3) four small $O(\log m)
\times O(\log n)$ submatrix queries in $M$ for the ranges $(R_i,C_j)$,
$i,j \in \{p,s\}$. We find the maximum
in each of these six ranges and return the maximum of the six
values. 

We find the maximum of each of the small  $O(\log m) \times
O(\log n)$ ranges of $M$ in $O(\log m +\log n )$ time using the SMAWK
algorithm. The maximum in the submatrix of $M'$ is found using $T_h$ as follows (the maximum in the submatrix of $M''$ is found similarly using $T_v$). Notice that 
$R'$ is the disjoint union of $O(\log m)$ canonical nodes of $T_h$. 
For each such canonical node $u$, we use binary-search on $u$'s list of breakpoints $\{j_1,j_2,\ldots, j_k\}$ to  find  the  set $\{j_a,j_{a+1},\ldots, j_b\}$
of $u$'s breakpoints that are fully contained in $C$. Although this
binary-search can take $O(\log m)$ time for each canonical node, using fractional cascading,
the searches on {\em all} canonical nodes take only $O(\log m)$
time and not  $O(\log^2 m)$.  
The maximum in all rows of $u$ and all columns between $j_a$ and $j_b$
is found by one query to the RMQ array $A^u$ in $O(1)$ time. Over
all canonical nodes this takes  $O(\log m)$ time. 

The  columns of $C$ that are to the left of $j_a$ all have their
maximum element in row $r(j_a)$ of $M'$ (that is, in one of $O(\log m)$ rows of $M$) . Similarly, the columns of $C$ that
are to the right of $j_b$ all have their maximum element in row
$r(j_{b+1})$ of $M'$. This means we have two rows of $M'$, $r(j_a)$ and $r(j_{b+1})$,
where we need to search for the maximum. We do this only after we have
handled all canonical nodes. That is, after we handle all canonical
nodes we have a set $A = a_1, a_2, \dots$ of $2 \log m$ rows of $M'$ in which we still need to find the
maximum. We apply the same procedure on $T_v$ which gives us a set $B
= b_1,b_2, \dots$
of $2 \log n$ columns of $M''$ in which we still have to find the
maximum.
Note that we only need to find the maximum among the elements of $M$ that
lie in rows corresponding to a row in $A$ and in columns
corresponding to a column in $B$.
This amounts to finding the maximum of the $O(\log m) \times
O(\log n)$  matrix $\bar M$, with $\bar M_{ij}$ being the maximum
among the elements of $M$ in the intersection of the $x$ rows corresponding to row $a_i$
of $M'$, and of the $x$ columns corresponding to column $b_j$ of $M''$.

An argument similar to the one in Lemma~\ref{lemma:3.1improvement}
shows that $\bar M$ is Monge. Therefore we can find its maximum
element using the SMAWK algorithm. We claim that each element of
$\bar M$ can be computed in $O(1)$ time, which implies that SMAWK
finds the maximum of $\bar M$ in $O(x)$ time. 

It remains to show how to compute an element of $\bar M$ in constant
time.
Recall from the proof of Lemma~\ref{lemma:3.1improvement} that $M$ is
partitioned into $x$-by-$n$ matrices $M^i$. During the preprocessing stage, for each $M^i$ we compute
and store its upper envelope, and an RMQ array over the maximum
elements in each interval of the envelope (similar to the array
$A^u$). Computing the upper envelope takes $O(x\log n)$ time by
incrementally adding one row at a time and using binary search to
locate the new breakpoint contributed by the newly added row. Finding
the maximum within each interval of the upper envelope can be done in
$O(x\log n)$ time using the tree $\mathcal B$. We store the upper
envelope in an atomic heap~\cite{FredmanW94}, which supports
predecessor searches in constant time provided $x$ is
$O(\log n)$. Overall the
preprocessing time is $O(m \log n)$, and the space is $O(m)$.
We repeat the same preprocessing on the transpose of $M$.

Now, given a row $a_i$ of $M'$ and column $b_j$ of $M''$, let
$[c_a,c_b]$ be the range of $x$ columns of $M$ that correspond to
$b_j$.
We search in
constant time for the successor $c_{a'}$ of $c_a$ and for the predecessor
$c_{b'}$ of $c_b$ in the upper envelope of $M^{a_i}$. We use
the RMQ array to find in $O(1)$ time the maximum element $y$ among elements in all rows
of $M$ corresponding to $a_i$ and columns in the range
$[c_{a'},c_{b'})$. The maximum element in columns $[c_a,c_{a'})$ and $[c_{b'},c_b]$  is contributed
by  two known rows $r_1,r_2$. We repeat the symmetric process for
the transpose of $M$, obtaining a maximum element $y'$, and two
columns $c_1,c_2$. $\bar M_{a_i,b_j}$ is the maximum among six values: $y,y'$ and
the four elements $M_{r_1c_1}, M_{r_1c_2}, M_{r_2c_1}, M_{r_2c_2}$.
\qed \end{proof}
%

 Notice that in the above proof, in order to obtain an element of
 $\bar M$ in constant time, we loose the $O(\log\log m)$ speedup in
 the construction time. This is because we found the upper envelope of
 each $M^i$.  To get the $O(\log\log m)$ speedup we can 
obtain an element of $\bar M$ in $O(x^\varepsilon)$ time using the
micro data structure.

\begin{corollary}\label{cor:3.2improvement}
Given a $m\times n$ Monge matrix,  one can construct, in $O((m+n)(\log  n+
\log m)/\log\log m)$ time, a data structure of size $O(m+n)$ that reports the maximum entry in a query
submatrix in $O((\log m+\log n)^{1+\varepsilon})$ time for any fixed $0<\varepsilon<1$. 
\end{corollary}

\paragraph{\bf A linear-space subcolumn data structure for partial matrices.} 

We next claim that the bounds of Lemma~\ref{lemma:3.1improvement} for
TM matrices also  apply to {\em partial} TM
 matrices. The reason is that we can efficiently turn any
partial TM matrix $M$ into a full TM
matrix by implicitly filling appropriate constants instead of the blank
entries.

\begin{lemma}\label{lemma:filltheblanks}
The blank entries in an $m\times n$ partial Monge matrix $M$ can be implicitly replaced in $O(m+n)$ time so that $M$ becomes Monge and each $M_{ij}$ can be returned in $O(1)$ time.
\end{lemma}
\begin{proof}
Let $s_i$ (resp. $t_i$) denote the index of the leftmost (resp. rightmost) column that is defined in row $i$. 
Since the defined (non-blank) entries of each row and column are continuous we have that the sequence $ s_1,s_2,\ldots,s_m$ starts with a non-increasing prefix $s_1\ge s_2\ge \ldots \ge s_a$ and ends with a non-decreasing suffix $s_a\le s_{a+1}\le \ldots \le s_m$. 
Similarly, the sequence $t_1,t_2,\ldots,t_n$ starts with a non-decreasing prefix $ t_1\le t_2\le \ldots \le t_b$ and ends with a non-increasing suffix $t_b\ge t_{b+1}\ge \ldots \ge t_m$. 

We partition the blank region of $M$ into four regions: (I) entries that are above and to the left of $M[i, s_i]$ for $i=1,\ldots,a$, 
(II) entries that are below and to the left of $M[i, s_i]$ for $i=a+1,\ldots,m$, 
(III) entries that are above and to the right of $M[i ,t_i]$ for $i=1,\ldots,b$, 
(IV) entries that are below and to the right of $M[i ,t_i]$ for $i=b+1,\ldots,n$.
We first describe how to replace all entries in region I to make them non-blank and
obtain a valid partial Monge matrix (whose blank entries are only in regions II, III, and IV). The remaining regions are handled in a similar manner, one after the other.

We describe our method for filling in the blank entries in region I in two steps. In the first step we show how to implicitly fill in the blanks in a lower right triangular Monge matrix so that each filled blank entry can be computed in $O(1)$ time. By a lower right triangular Monge matrix we mean a partial Monge square matrix with $n$ rows and columns, such that, for all $1\leq i\leq n$, $s(i) = n-i+1$. In the second step we explain that any $m$-by-$n$ partial Monge matrix whose blank entries are in region I can be turned into a lower right triangular Monge matrix with at most $m+n$ rows and columns. The only operations used in the transformation are duplicating rows, duplicating columns, and turning elements into blanks. We will show an $O(m+n)$ procedure for computing two tables. One specifying, for each row index $1\leq i\leq m$, the corresponding row index in the larger $O(m+n)$ triangular matrix. The second is an analogous table for the columns indices. The lemma then follows for the blank entries in region I. The other regions are treated by reducing to the region I case, one after the other.

We now describe how to fill in the blank regions in a lower right triangular Monge matrix. Let $W$ denote the largest absolute value of any non-blank entry in $M$ (We can find $W$ by applying the
algorithm of Klawe and Kleitman~\cite{KK89}). 
Intuitively, we would like to make every $M[i,j]$ in the upper left triangle very large. However,
we cannot simply assign the same large value to all of them, because then the Monge
inequality would not be guaranteed to hold if more than one of the four considered elements
belongs to the replaced part of the matrix. 
A closer look at all possible cases shows that setting all the entries of each diagonal to the same value does work. 
More precisely, we replace the blank element $M[i,j]$ with $3W[2(n-i-j)+1]$. Thus, each element in the first diagonal off the main diagonal ($i+j=n$) is set to $3W$, the elements of the second diagonal off the main diagonal are set to $9W$, etc. Note that the maximum element in the resulting matrix is $O(nW)$. 
To prove that the resulting new matrix $M'$ is Monge, it suffices, by Proposition~\ref{prop:adjacent condition}, to show that, for all $1\leq i,k < n$,  $M'[i,k]+M'[i+1,k+1]-M'[i,k+1]-M'[i+1,k] \geq 0$.
To this end we consider the following cases:
\begin{enumerate}
\item $i+k>n$, so all $M[i,k],M[i+1,k+1],M[i,k+1],M[i+1,k]$ are non-blank, and the inequality holds because $M$
is partial Monge.
\item $i+k=n$, so $M[i,k]$ is blank and $M[i+1,k+1],M[i,k+1],M[i+1,k]$ are non-blank. Then\\
$M'[i,k]+M'[i+1,k+1]-M'[i,k+1]-M'[i+1,k] = 3W + M'[i+1,k+1]-M'[i,k+1]-M'[i+1,k] \geq 3W-3W = 0$.
\item $i+k=n-1$, so $M[i,k],M[i,k+1],M[i+1,k]$ are blank, and $M[i+1,k+1]$ is non blank. Then\\
$M'[i,k]+M'[i+1,k+1]-M'[i,k+1]-M'[i+1,k] = 9W + M[i+1,k+1] - 3W-3W \geq 3W - W \geq 0$.
\item $i+k<n-1$, so all $M[i,k],M[i+1,k+1],M[i,k+1],M[i+1,k]$ are blank. Then,\\ 
$M'[i,k]+M'[i+1,k+1]-M'[i,k+1]-M'[i+1,k] = 0$.
\end{enumerate}
Hence the new matrix $M'$ is indeed Monge.

Next, we describe how to turn any $m$-by-$n$ partial Monge matrix $M$ whose blank entries are in region I into a slightly larger lower right triangular matrix $M'$. This is done by duplicating rows or columns of $M$ and replacing by blanks a nonempty prefix in all but a single copy. Thus, each row $r$ (column $c$) of $M$ has exactly one appearance in $M'$ in which no elements are replaced by blanks. We say that $r$ ($c$) is mapped to this appearance in $M'$.   Propositions~\ref{prop:undef} and~\ref{prop:duplicate} guarantee that $M'$ is partial Monge. 
For ease of presentation we describe the process as if we actually transform the $M$ into $M'$. 

The assumption that the blank entries are in region I implies that $s_1 \geq s_2 \geq \cdots \geq s_m$ and that $t_1=t_2=\cdots = t_m$. 
We first guarantee that the $s_i$'s are strictly decreasing. We do this by iterating through the $s_i$'s. If $s_i = s_{i-1}$, we duplicate the column $s[i]$ of $M$, make $M[i-1,s[i]]$ blank, and mark the column currently at index $s[i]$ as a duplicate (the index of this column might change later if columns with smaller indices will be duplicated). This column duplication has the effect of increasing by 1 all $s[j]$'s for $j<i$. At the end of this process we construct a table $c[\cdot]$ which keeps track of the mapping of columns of $M$ to $M'$ by recording for each non-duplicate row its original index in $M$ and its index after this process. Clearly, computing $c[\cdot]$ and updating the $s_i$'s can be done in $O(m+n)$ time without actually duplicating the columns.

We may now assume that the $s_i$'s are strictly decreasing, and we use $n$ to denote the number of columns of $M$ after the transformation that ensured the strict monotonicity. 
We use a table $r[\cdot]$ to keep track of the mapping from rows of $M$ to rows of $M'$. 
For convenience, we define $s_0 = n+1$, and $r[0] =0$. We iterate through the sequence $s_1,s_2, \dots, s_m$. 
We add to $M'$ $s_{i-1}-s_i$ copies of row $i$ of $M$, and, for $j=1,2, \dots, s_{i-1}-s_i-1$, replace the prefix of length $j$ from the $j$'th copy by blanks, so only the last copy remains unchanged. We therefore set $r[i]$ to $r[i-1]+s_{i-1}-s_i$. Clearly, we can compute the table $r[\cdot]$ in $O(m)$ time without actually constructing $M'$.

Finally, to obtain the value with which the blank entry at $M[i,j]$ should be replaced when converting $M$ into a full Monge matrix, we return $3W[2(n-r[i]-c[j])+1]$.

Regions II, III, and IV can be handled symmetrically to region I. To handle undefined entries in region II, we 
implicitly reverse the order of the rows and negate all the elements of the matrix. It is easy to verify that the resulting matrix is Monge with undefined entries in region I. We then implicitly fill in the undefined values using the method described above, negate all the elements and revert the order of rows to its original order. The transformation for region III is reversing the order of columns and negating all elements, and the transformation for region IV is reversing the order of both rows and columns.
Note that to make $M$
full Monge we first need to fill the blanks in region I, then calculate the new value of $W$
and fill the blanks in region II accordingly, and so on.
\qed \end{proof}

The above lemma means we can (implicitly) fill the black entries in $M$ so that $M$ is a {\em full} TM
matrix. We can therefore apply the data structure of Lemma~\ref{lemma:3.1improvement}. 
Note that the maximum element in a query (a column $\pi$ and a range of rows $R$) might now appear in one of the previously-blank entries. 
This is easily overcome by first restricting $R$ to the defined entries in the column  $\pi$ and only then querying the data structure of Lemma~\ref{lemma:3.1improvement}.

\paragraph{\bf A linear-space submatrix data structure for partial matrices.} 
Given a partial matrix $M$, the above simple trick of replacing
appropriate constants instead of the blank entries does not work for
submatrix queries because the defined (i.e., non-blank) entries in a
submatrix do not necessarily form a submatrix. Instead, we need a more
complicated construction, which yields the following theorem.

\begin{theorem}\label{theorem:3.4improvement}
Given a $m\times n$ partial Monge matrix,  one can construct, in $O((m + n) 
\log (m+n))$ time, a data structure of size $O(m  + n )$ that reports the maximum entry in a query
submatrix in  $O((\log m+\log n) \alpha(m+n))$~time.
\end{theorem}

\begin{proof}
As before, we partition $M$ into $m/x$ matrices $M^1, M^2, \ldots,
M^{m/x}$, where $x=\log m$ and $M^i$ is an $x\times n$ matrix composed
of $x$ consecutive rows of $M$. We wish to again define the $(m/x)
\times n$ matrix $M'$ such that $M'_{ij}$ is equal to the maximum
entry in column $j$ of $M^i$. However,  it is now possible that some
(or all) of the entries in column $j$ of $M^i$ are undefined. We
therefore define $M'$ so that $M'_{ij}$ is equal to the maximum entry
in column $j$ of $M^i$ only if the entire  column $j$ of $M^i$ is
defined. Otherwise,  $M'_{ij}$ is undefined. We also define the sparse matrix $S'$ so that
$S'_{ij}$ is undefined if column $j$ of $M^i$ is either entirely
defined or entirely undefined. Otherwise, $S'_{ij}$  is equal to the
maximum entry among all the defined entries in column $j$ of $M^i$. 

Using a similar argument as before, it is easy to show that $M'$ is also a partial Monge matrix. 
The matrix $S'$, however, is not partial Monge, but it is a sparse
matrix with at most two entries per column. It has additional
structure on which we elaborate in the sequel.

We begin with $M'$. As before, we cannot afford to store $M'$
explicitly. Instead, we use the micro data structure on $M^1,  \ldots,
M^{m/x}$ (after implicitly filling the blanks in $M$ using
Lemma~\ref{lemma:filltheblanks}). This time we use $r=1$ and so the
entire construction takes   
$O((m/x) x \log n/\log r)=O(m\log n)$ time and $O(m)$ space, after
which we can retrieve any entry of $M'$  in $O(pred(x,n))$ time.  
We then build a similar data structure to the one we used in
Theorem~\ref{theorem:3.2improvement}.  That is, we build $T_h$ on
$M'$, and for $\mathcal B$ we use the data structure of
Lemma~\ref{lemma:3.1improvement} applied to the transpose of $M$
(after implicitly filling the blanks).  
$\mathcal B$'s construction therefore requires $O(n(\log  m+ \log
n)/\log\log n)$ time and $O(n)$ space.  

After constructing $\mathcal B$, constructing $T_h$ (along with the  RMQ arrays $A^u$) on $M'$ is done bottom up. This time, since $M'$ is partial Monge, each node of $T_h$ can contribute more than one new breakpoint. However, as we show in Section~\ref{sec:upperbound} (Theorem~\ref{lemma:partial-breakpoints}),
a node whose subtree contains $k$ leaves (rows) can contribute at most $O(k)$ new breakpoints.
Each new breakpoint can be found in $O(\log n)$ time via binary search. Summing $O(k\cdot  \log n \cdot pred(x,n))$ over all 
$m/x$ nodes
of $T_h$ gives $O((m/x) \log(m/x) \cdot \log n \cdot pred(x,n))) =
O(m\log n)$  time  and $O(m/x\cdot \log(m/x))=O(m)$
space. Notice we use atomic heaps here to get $pred(x,n)=O(1)$.

Similarly to what was done in Theorem~\ref{theorem:3.2improvement},
we repeat the entire preprocessing with the transpose of $M$ (that is,
we construct $T_v$ on the columns of the $m \times (n/\log n)$ matrix
$M''$, along with the RMQ data structures, and also construct the
corresponding sparse matrix $S''$). This takes $O(n\log m)$
time  and $O(n)$ space. 

We now describe how to answer a submatrix query with row range $R$ and
column range $C$. Let $R', R_s, R_p, C', C_s, C_p$ be as in
Theorem~\ref{theorem:3.2improvement}. 
The submatrix query $(R,C)$ can be covered by the following:
(1) a submatrix query $(R',C)$ in $M'$, 
(2) a submatrix query $(R',C)$ in $S'$, 
(3) a submatrix query $(R,C')$ in $M''$, 
(4) a submatrix query $(R,C')$ in $S''$, 
and (5) four small $O(\log m)
\times O(\log n)$ submatrix queries in $M$ for the ranges $(R_i,C_j)$,
$i,j \in \{p,s\}$.
We return the overall maximum among the maxima
in each of these queries.

We already described how to handle the queries in items (1), (3), and
(5)  in the proof of Theorem~\ref{theorem:3.2improvement}. The only
subtle difference is that in Theorem~\ref{theorem:3.2improvement} we
used the SMAWK algorithm on $O(\log m) \times O(\log n)$ Monge
matrices while here we have partial Monge matrices. We therefore use
the Klawe-Kleitman algorithm~\cite{KlaweK90} instead of SMAWK which
means the query time is $O((\log m+\log n) \alpha(n))$ and not
$O(\log m+\log n)$.  

We next consider the query to $S'$. The query to $S''$ is handled in a
similar manner. 
Recall from the proof of Lemma~\ref{lemma:filltheblanks} the structure
of a partial matrix $M$. Let $s_i$ (resp. $t_i$) denote the index of the leftmost (resp. rightmost) column that is defined in row $i$. 
Since the defined (non-blank) entries of each row and column are continuous we have that the sequence $ s_1,s_2,\ldots,s_m$ starts with a non-increasing prefix $s_1\ge s_2\ge \ldots \ge s_a$ and ends with a non-decreasing suffix $s_a\le s_{a+1}\le \ldots \le s_m$. 
Similarly, the sequence $t_1,t_2,\ldots,t_n$ starts with a
non-decreasing prefix $ t_1\le t_2\le \ldots \le t_b$ and ends with a
non-increasing suffix $t_b\ge t_{b+1}\ge \ldots \ge t_m$. 
See Fig.~\ref{fig:partial} for an illustration.
It follows that the defined entries of $S'$ can be partitioned into
four sequences, such that the row and column indices in each sequence
are monotone. We focus on
one of these monotone sequences in which the set of defined entries is
in coordinates $(r_1,c_1),(r_2,c_2),\ldots$ such that
$r_{i+1} \ge r_i$ and $c_{i+1}\le c_i$. The other monotone sequences
are handled similarly. 
Notice that any query range
that includes $(r_i,c_j)$ and $(r_j,c_j)$ for some $i<j$ must include
entries $(r_k,c_k)$ for  all $i<k<j$. Given a range query $(R,C)$, we
find in $pred(n,n)$ time the interval $[i_1,i_2]$ of indices  that are inside $R$.
Similarly,
we find the interval $[i'_1,i'_2]$ of indices  that are inside $C$.
We can then  use a (1-dimensional) RMQ data
structure on the $O(n)$ entries in this sequence 
to find the maximum element in the intersection of these two
ranges in $O(1)$
time. Overall, handling the query in $S'$ takes $pred(n,n) = O(\log\log
n)$ time.

To conclude the proof of Theorem~\ref{theorem:3.4improvement}, notice that our data structure requires $O(m+n)$ space, is constructed in $O(m\log n  + n\log n / \log\log n+n\log m + n\log n)$ time which is $O(n\log n)$, and has $O((\log m+\log n) \alpha(n))$  query time.
\qed \end{proof}
%

\noindent Finally,  for the same reasons leading to Corollary~\ref{cor:3.2improvement} we can get a $\log\log m$ speedup in the construction-time with a $\log^\varepsilon n$ slowdown in the query-time.

\begin{corollary}\label{cor:3.4improvement}
Given a $m\times n$ partial Monge matrix,  one can construct, in $O((m + n) 
\log (m+n)/\log\log m)$ time, a data structure of size $O(m  + n )$ that reports the maximum entry in a query
submatrix in  $O((\log m+\log n)^{1+\varepsilon} \alpha(m+n))$ time for any fixed $0<\varepsilon<1$.
\end{corollary}

\section{The Complexity of the Upper Envelope of a Totally Monotone Partial 
  Matrix}\label{sec:upperbound}
In this section we prove the following theorem, stating that the number of breakpoints of an $m\times n$ TM partial  matrix is only $O(m)$. 

\begin{theorem}\label{lemma:partial-breakpoints}
Let $M$ be a partial $m\times n$ matrix in which the defined entries in each row
and in each column are
contiguous.
If $M$ is TM (i.e., for all $i<j, k<\ell$ where
$M_{ik},M_{i\ell},M_{jk},M_{j\ell}$ are all defined, $M_{ik} \leq M_{jk}
\implies M_{i\ell} \leq M_{j\ell}$), then the upper envelope has
complexity $O(m)$.
\end{theorem}

The proof relies on a decomposition of $M$ into {\em staircase}
matrices. A partial matrix is staircase if the defined entries in its rows either all  begin in the first column or all end in the last column. 
It is well known (cf.~\cite{AggarwalK90}) that by cutting
$M$ along columns and rows, it can be decomposed into staircase
matrices $\{M_i\}$ such that each row is covered by at most three matrices,
and each column is covered by at most three  matrices.
For completeness, we  describe such a decomposition below.

\begin{lemma}
A partial matrix $M$ can be decomposed into staircase
matrices $\{M_i\}$ such that each row is covered by at most three matrices,
and each column is covered by at most three  matrices.
\end{lemma}
\begin{proof}

Let $s_i$ and $t_i$ denote the smallest and largest column index
in which an element in row $i$ is defined, respectively. 
The fact that the defined entries of $M$ are contiguous in both rows
and columns implies that the sequence $s_1, s_2, \dots, s_m$ consists of a
non-increasing prefix and a non-decreasing suffix. Similarly, the 
sequence $t_1, t_2, \dots, t_m$ consists of a
non-decreasing prefix and a non-increasing suffix. 
It follows that the rows of $M$ can be divided into three ranges - 
a prefix where $s$ is non-increasing and $t$ is non-decreasing, an infix where
both $s$ and $t$ have the same monotonicity property, and a suffix
where $s$ is non-decreasing and $t$ is non-increasing.
The defined entries in the prefix of the rows can be divided into two
staircase matrices by splitting $M$ at $t_1$, the largest column where the
first row has a defined entry. 
Similarly, the defined entries in the suffix of the rows can be divided into two
staircase matrices by splitting it at $t_m$, the largest column where the
last row has a defined entry. 
The defined entries in the infix of the rows form a double staircase
matrix. It can be broken into staircase matrices by dividing along
alternating rows and columns as shown in Figure~\ref{fig:partial}. 

\begin{figure}[h!]
   \centering
   \includegraphics[scale=0.6]{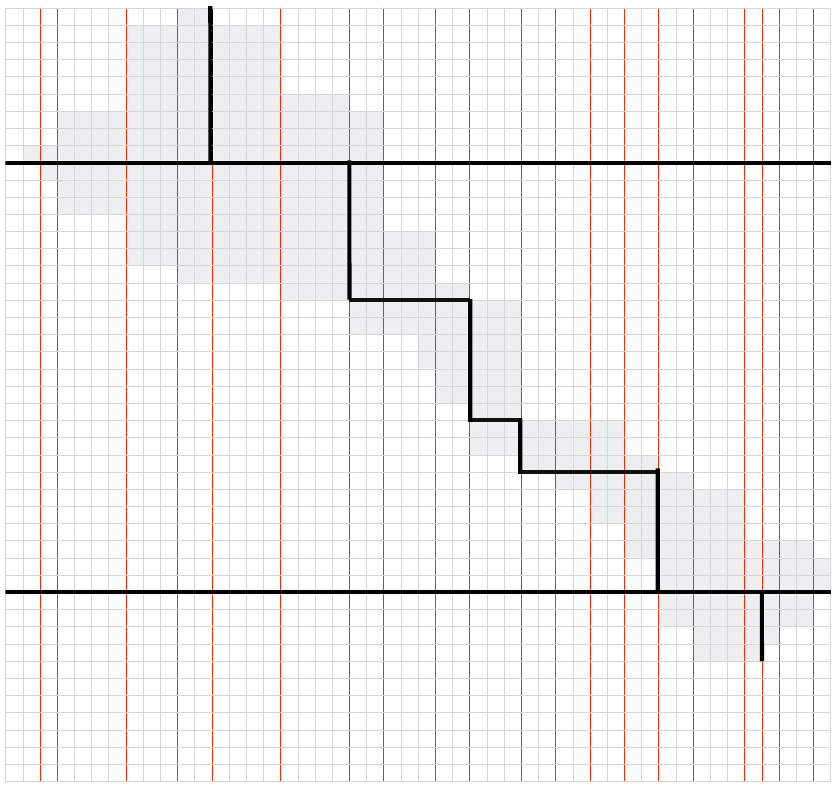}
   \caption{Decomposition of a partial matrix into staircase matrices
     (defined by solid thick black lines) and into blocks of consecutive
     columns with the same defined entries (indicated by thin vertical
     red lines).}
  \label{fig:partial}
 \end{figure}

It is easy to verify that, in the resulting decomposition, each row
is covered by at most two staircase matrices, and each column is covered
by at most three staircase matrices.
Also note that  every set of consecutive
columns whose defined elements are in exactly the same set of rows are covered
in this decomposition by the same three row-disjoint staircase matrices. \qed
\end{proof}

We next prove the fact that, if $M$ is a TM staircase
matrix with $m$ rows, then the complexity of its upper envelope is
$O(m)$. 

\begin{lemma}\label{lemma:2m:staircase}
The number of breakpoints in the upper envelope of an $m\times n$ TM staircase matrix is at
most $2m$.
\end{lemma}
\begin{proof}
We focus on the case where the defined entries of all rows begin in the first column and end in non-decreasing columns. In other words, for all $i$, $s_i$=1 and $t_i\le t_{i+1}$. The other cases are symmetric. 

A breakpoint is a situation where the maximum in column $c$ is at row
$r_1$ and the maximum in column $c+1$ is at a different row $r_2$.
We say that $r_1$ is the departure row of the breakpoint, and $r_2$ is
the  entry row of the breakpoint.
There are two types of breakpoints: decreasing ($r_1 < r_2$), and
increasing  ($r_1 > r_2$).
We show that 
(1) each row can be the entry row of at most one decreasing breakpoint, and (2) each row can be the
departure row of at most one increasing breakpoint. 
\begin{enumerate}
\item[(1)] Assume that row $r_2$ is an entry row of two decreasing
  breakpoints:  One is the pair of entries $(r_1,c_1),(r_2,c_1+1)$ and
  the other is the pair  $(r_3,c_2),(r_2,c_2+1)$. We know that
  $r_1<r_2$, $r_3<r_2$, and wlog $c_2>c_1+1$. 
Since the maximum in column $c_1+1$ is in row $r_2$, we have
$M_{r_3,c_1+1} < M_{r_2,c_1+1}$.
However, since the maximum in column $c_2$ is in row $r_3$, we have
$M_{r_3,c_2}  > M_{r_2,c_2}$, contradicting the total monotonicity of
$M$. Note that $M_{r_2,c_2}$ is defined since $M_{r_2,c_2+1}$ is defined.

\item[(2)]
Assume that row $r_1$ is a departure row of two increasing breakpoints:
One is the pair of entries $(r_1,c_1),(r_2,c_1+1)$ and the other is
the pair  $(r_1,c_2),(r_3,c_2+1)$. We know that $r_1>r_2$ and
$r_1>r_3$. 
Since the maximum in column $c_1$ is in row $r_1$, we have
$M_{r_2,c_1} < M_{r_1,c_1}$.
However, since the maximum in column $c_1+1$ is in row $r_2$, we have
$M_{r_2,c_1+1}  > M_{r_1,c_1+1}$, contradicting the total monotonicity of
$M$. Note that $M_{r_1,c_1+1}$ is defined since $M_{r_1,c_2}$ is defined.
\end{enumerate} 
\qed \end{proof}

Using Lemmas~\ref{lemma:partial-breakpoints} and~\ref{lemma:2m:staircase} we can now  complete the proof of Theorem~\ref{lemma:partial-breakpoints}.
Let $bp(M_i)$ denote the number of breakpoints in the upper envelope
of $M_i$. 
Let $m_i$ denote the number of rows in $M_i$.
Since each row appears in at most three $M_i$s, $\sum_i m_i =
O(m)$.
The total number of breakpoints in the envelopes of all of  $M_i$s is
$O(m)$ since 
$\sum_i bp(M_i) = \sum_i O(m_i) = O(m)$.

Consider now a partition of $M$ into rectangular blocks $B_j$ defined by maximal
sets of contiguous columns whose defined entries are at the same set
of rows. There are $O(m)$ such blocks.
The upper envelope of $M$ is just the concatenation of the upper
envelopes of all the $B_j$'s. Hence, $bp(M) = \sum_j bp(B_j) + O(m)$ (the
$O(m)$ term accounts for the possibility of a new breakpoint between every two
consecutive blocks). Therefore, it suffices to bound $\sum_j bp(B_j)$.

Consider some block $B_j$. As we mentioned above, the columns of $B_j$
appear in the same three row-disjoint staircase matrices $M_1,M_2,M_3$ in the decomposition of
$M$. The column maxima of $B_j$ are a subset of the column maxima of
$M_1,M_2,M_3$. Assume wlog that the indices of rows covered by $M_1$ are smaller than
those covered by $M_2$, which are smaller than those covered by $M_3$. 

The breakpoints of the upper envelope of $B_j$
are either breakpoints in the envelope of $M_1,M_2,M_3$, or
breakpoints that occur when the maxima in consecutive columns of $B_j$
originate in different $M_i$. However, since $B_j$ is a (non-partial)  TM matrix, its column maxima are
monotone. So once a column maximum originates in $M_i$, no maximum in
greater columns will ever originate in $M_j$ for $j<i$. It follows
that the number of breakpoints in $B_j$ that are not breakpoints of
$M_1,M_2,M_3$ is at most two. Since there are  $O(m)$ blocks, 
$\sum_j bp(B_j) \leq \sum_i bp(M_i) + O(m) = O(m)$. This completes the proof of Theorem~\ref{lemma:partial-breakpoints}.

\section{Constant Query-Time Data Structures}\label{sec:fast-query}
In this section we present our data structures that improve the query time to $O(1)$ at the cost of an $n^\varepsilon$ factor in the construction time and space for any constant $0<\varepsilon<1$.

We use the following micro data structure that slightly modifies the one of Lemma~\ref{lemma:micro}. 

\begin{lemma} [another micro data structure]\label{lemma:micro2}
Given a TM matrix of size $x\times n$, one can construct in $O(xn^\varepsilon\varepsilon^{-1} )$ time and space a data structure that given a query column can report the maximum entry in the entire column in 
$O(\log(\varepsilon^{-1}))$ time for any $1>\varepsilon \ge \log\log n/\log n$. 
\end{lemma}
\begin{proof}
Recall that the data structure of Lemma~\ref{lemma:micro}, for $r=1$, finds in  $O(x \log n)$ time a set of $O(x)$ values (breakpoints) in the range $\{1,\ldots,n\}$. A query is performed in $O(pred(x,n))$ time using a standard predecessor data structure on these $O(x)$ values. Since now we  can allow an $n^\varepsilon$ factor we  use a non-standard  predecessor data structure with faster  $O(\varepsilon^{-1})$  query-time. We now describe this data structure.

Consider the complete tree of degree $n^{\varepsilon}$ over the leaves $\{1,\ldots,n\}$. We do not store this entire tree. We only store the leaf nodes corresponding to the $O(x)$ existing values and all ancestors of these leaf nodes. Since the height of the  tree is $O(\varepsilon^{-1})$ we store only $O(x\varepsilon^{-1})$ nodes. At each such node we keep two arrays, each of size $n^{\varepsilon}$. The first array stores all children pointers (including null-children). The second array stores for each child $u$ (including null-children) the value $pred(u)=$  the largest existing leaf node that appears before $u$ in a preorder traversal of the tree.

The $y= O(x\varepsilon^{-1})$ nodes are stored in a hash table. We use the static deterministic hash table of Hagerup et al.~\cite{Hash} that is constructed in $O(y \log y) = O(x\varepsilon^{-1} \log (x\varepsilon^{-1}))$ worst case time and can be queried in $O(1)$ worst case time.
Upon query, we binary-search (using the hash table) for the deepest node $v$ on the root-to-query path whose child $u$ on the root-to-query path is null.  To find the predecessor we use $v$'s second array and  return $pred(u)$.

The total construction time is $O(x \log n +  xn^{\varepsilon}\varepsilon^{-1} +x\varepsilon^{-1} \log (x\varepsilon^{-1}))$ which is $O(xn^{\varepsilon}\varepsilon^{-1})$ since we assume $\varepsilon \ge \log\log n/\log n$. The query time is $O(\log(\varepsilon^{-1}))$  since we binary-search on a path of length $\varepsilon^{-1}$ and each lookup takes $O(1)$ time using the hash table.
\qed \end{proof}

\noindent We use the above micro data structure to obtain the following  data structure.

\begin{lemma}\label{lem:skewedDS}
Given a TM $x\times n$ matrix, one can construct in $O(x^3n^\varepsilon\varepsilon^{-1} )$ time a data
structure of size $O(x^3n^\varepsilon\varepsilon^{-1} )$ that can report the maximum entry in a query
column and a contiguous range of rows in $O(\log(\varepsilon^{-1}))$ time.
\end{lemma}

\begin{proof}
For each of the $O(x^2)$  row intervals, construct the data structure of Lemma~\ref{lemma:micro2}.
\qed \end{proof}

\paragraph{\bf A constant-query subcolumn data structure.} 

\begin{lemma}\label{lemma:3.1fast-query}
Given a TM matrix of size $m\times n$, one can construct, in 
$O(mn^{\varepsilon} 
\varepsilon^{-2} ) = O(n^{1+\varepsilon} \varepsilon^{-2} )$
time and space a data structure that can report the maximum entry in a query
column and a contiguous range of rows in $O(\varepsilon^{-1} \log(\varepsilon^{-1}))$ time.
\end{lemma}

\begin{proof}
The first idea is to use a degree-$x$ tree, with $x = m^{\varepsilon/4}$, instead of the binary tree $T_h$.
The height of the  tree is $O(\log m / \log x) = O(\varepsilon^{-1})$.
The leaves of the tree correspond to individual rows of $M$.
For an internal node $u$ of this tree, whose children are $u_1,u_2,
\dots, u_x$ and whose subtree contains $k$ leaves (i.e., $k$ rows), recall that $M^u$ 
is the $k \times n$ matrix  defined by these $k$ rows and all columns. 
Let  $\hat{M}^u$ be the $x\times n$ matrix whose $(i,j)$ element
is the maximum in column $j$ among the rows of $M^{u_i}$. In other words, $\hat{M}^u(i,j) = \max_{\ell}
{M}^{u_i}(\ell,j)$. 

Working bottom up, for each internal node $u$, instead of explicitly storing the matrix $M^u$ (whose size is $O(kn)$), we build the $O(kn^{\varepsilon}\varepsilon^{-1})$-sized micro data structure of Lemma~\ref{lemma:micro2}  over the $k$ rows of $M^u$. 
This way, any element $\hat{M}^u(i,j)$ can be obtained in $O( \log(\varepsilon^{-1}))$ time by querying
the data structure of $u_i$. Once we can obtain each $\hat{M}^u(i,j)$ in $O( \log(\varepsilon^{-1}))$  time, we use this to 
 construct the  data structure of Lemma~\ref{lem:skewedDS} over the $x=m^{\varepsilon/4}$ rows of $\hat{M}^u$.

Constructing the micro data structure of Lemma~\ref{lemma:micro2} for an internal node with $k$ leaf descendants takes  $O(kn^{\varepsilon}\varepsilon^{-1})$ time and  space. Summing this over all internal nodes in the tree, the total
construction takes $O(m n^{\varepsilon}\varepsilon^{-2} )$ time and space. 
After this, we construct the Lemma~\ref{lem:skewedDS} data structure  for each internal node
but we use $\varepsilon/2$ and not $\varepsilon$ so the construction takes 
$O(x^3n^{\varepsilon/2} \varepsilon^{-1} \cdot \log(\varepsilon^{-1})) = O(m^{3\varepsilon/4} n^{\varepsilon/2} \varepsilon^{-1}\log(\varepsilon^{-1}))$ time and space. The total
construction time over all $O(m/x) = O(m^{1-\varepsilon/4})$ internal nodes is thus $O(m^{1+\varepsilon/2}n^{\varepsilon/2} 
 \varepsilon^{-1}\log(\varepsilon^{-1}) ) = O(n^{1+\varepsilon}  \varepsilon^{-1}\log(\varepsilon^{-1}))$ time and space. 

We now describe how to answer a query. Given a query column and a row interval $I$, there is an induced set
of $O(\log m / \log x) = O(\varepsilon^{-1})$ canonical nodes. Each canonical node $u$ is responsible for a subinterval of  
$I$ (that includes all descendant rows of $u_i,u_{1+1}\ldots,u_j$ for some two children $u_i,u_j$ of $u$).  We find the maximum in this subinterval with one query to $u$'s Lemma~\ref{lem:skewedDS} data structure in  $O(\log(\varepsilon^{-1} ))$ time. The total query time is thus $O(\varepsilon^{-1}\log(\varepsilon^{-1} ))$.
\qed \end{proof}

\paragraph{\bf A constant query submatrix data structure.} 

\begin{theorem}\label{theorem:3.2fast-query}
Given a Monge matrix of size $m\times n$, one can construct, in 
$O(n^{1+\varepsilon}\varepsilon^{-3}\log(\varepsilon^{-1}))$ time and space, a data
structure that can report the maximum entry in a query
submatrix in $O(\varepsilon^{-2}\log(\varepsilon^{-1} ))$ time.
\end{theorem}
\begin{proof}
As in the proof of Lemma~\ref{lemma:3.1fast-query}, we construct
a degree-$x$ tree $T_h$ over the rows of $M$, with $x =
m^{\varepsilon/4}$. 
 Recall that $T_h$ 
includes, for each internal node $u$, (i) the data
structure of Lemma~\ref{lemma:micro2}, which enables queries to
elements of $\hat{M}^u$ in $O(\log(\varepsilon^{-1}))$ time, and (ii) the breakpoints of
all possible row intervals of the $x\times n$ matrix
$\hat{M}^u$. In addition to the breakpoints we store, for each of these
$O(x^2)$ intervals, a RMQ data structure over the maximum elements 
between breakpoints. The construction of those RMQ data structures is
described in the sequel. 

For each level $\ell>0$ of the $O(\varepsilon^{-1})$ levels of the
tree $T_h$ (the leaves of $T_h$ are considered to be at level 0), we construct the symmetric data structure of
Lemma~\ref{lemma:3.1fast-query} over the $(m/{x^{\ell-1}})\times n$ matrix
formed by the union of $\hat{M}^u$ over all level-$i$ nodes $u$ in
$T_h$.
We denote these data structures by $\mathcal
B_\ell$. 
Their construction takes total $O(\varepsilon^{-1}\cdot  mn^{\varepsilon} 
\varepsilon^{-2}  \cdot
\log(\varepsilon^{-1})) = O(n^{1+\varepsilon} \varepsilon^{-3}\log(\varepsilon^{-1}))$ time and space. 
For notational convenience we define $B_0$ to be equal to $B_1$.

We now describe how to construct the RMQ data structures for an
internal node $u$ at level $\ell$ of $T_h$ with children
$u_1, \dots, u_x$.  We describe how to construct the RMQ for the interval consisting
of all rows of $\hat{M}^u$. Handling the other intervals is similar.
We need to show how to list the maximum among the column maxima of $\hat{M}^u$ between every two
consecutive breakpoints of $\hat{M}^u$.
All the column maxima between any two consecutive breakpoints are
contributed by a single known child $u'$ of $u$. In other words, we
are looking for the maximum element in the range
consisting of a single row of $\hat{M}^u$ and the range of columns between the two
breakpoints. This maximum can be found by querying the $\mathcal B_\ell$ data structure
in $O(\varepsilon^{-1}\log(\varepsilon^{-1}))$ time.
There are $O(x)$ such queries for each of the $O(x^2)$ intervals at each of the $O(m/x)$
internal nodes. Therefore, the total construction
time of the RMQs is $O(m^{1+\varepsilon}\cdot \varepsilon^{-1}\log(\varepsilon^{-1}))$.  
This completes the description of our data structure. 

We finally discuss how to answer a query (a range in $M$ of  rows  $R$ and  columns $C$). A query induces a set of $O(\varepsilon^{-1})$
canonical nodes $u$. For a canonical node $u\in T_h$ and an
induced row interval $R_u$, we
use the list of breakpoints of $R_u$ in $\hat{M}^u$ to
identify the breakpoints that are fully contained in
$C$. This takes $O(\log(\varepsilon^{-1}))$ time. The maximum element in those columns is found by querying the RMQ
data structure of $R_u$ in $u$. In addition to that, there are at most
two column intervals $C'$ and $C''$ in $\hat{M}^u$
that intersect $C$ but are not fully contained in $C$. The maximum in
$C' \cap C$ and $C'' \cap C$ is contributed by two known children
$u',u''$ of $u$, respectively.
In other words, each of them is the maximum element in the range consisting of a single row of $\hat{M}^u$ and a range of columns. If $u$ is a level-$\ell$ node of the tree then we find them by two queries to $\mathcal B_\ell$: one for the row of $u'$ and columns $C'$ and one for the row of $u''$ and columns $C''$.
The total query time is thus 
$O(\varepsilon^{-1}\cdot \varepsilon^{-1}\log(\varepsilon^{-1}))= O( \varepsilon^{-2}\log(\varepsilon^{-1}))$.
\qed \end{proof}

\bibliographystyle{plain}

\end{document}